\documentclass[journal]{IEEEtran}[10pt]
\usepackage{array}
\usepackage{textcomp}
\usepackage{stfloats}
\usepackage{url}
\usepackage{verbatim}
\usepackage{cite}
\usepackage{color}
\usepackage{amssymb}
\usepackage{amsfonts}
\usepackage{amsmath}
\usepackage{amsthm}
\usepackage{mathrsfs}
\usepackage{bm}
\usepackage{soul}
\usepackage{verbatim}
\usepackage{autobreak} 
\usepackage{graphicx} 
\usepackage{float} 
\usepackage{subfigure} 
\usepackage{tabularx}
\usepackage{algpseudocode} 
\usepackage{enumerate}
\usepackage{cases}
\usepackage{algorithm} 
\usepackage{algpseudocode}

\usepackage[colorlinks,
            linkcolor=black,
            anchorcolor=black,
            citecolor=blue]{hyperref}

\usepackage{lettrine}

\newtheorem{proposition}{Proposition}

\begin{document}
\title{Secure Communication in Dynamic RDARS-Driven Systems}

\author{Ziqian Pei, Jintao Wang, Pingping Zhang, Zheng Shi, Guanghua Yang, Shaodan Ma
\thanks{Ziqian Pei, Jintao Wang, Pingping Zhang and Shaodan Ma are with the State Key Laboratory of Internet of Things for Smart City and the Department of Electrical and Computer Engineering, University of Macau, Macau 999078, China (e-mail: mc35150@um.edu.mo; wang.jintao@connect.um.edu.mo; yc17433@um.edu.mo; shaodanma@um.edu.mo).}
\thanks{Zheng Shi and Guanghua Yang are with the School of Intelligent Systems Science and Engineering, Jinan University, Zhuhai 519070, China (e-mails: zhengshi@jnu.edu.cn; ghyang@jnu.edu.cn).}
}



\maketitle
\thispagestyle{empty}
\begin{abstract}
In this letter, we investigate a dynamic reconfigurable distributed antenna and reflection surface (RDARS)-driven secure communication system, where the working mode of the RDARS can be flexibly configured. We aim to maximize the secrecy rate by jointly designing the active beamforming vectors, reflection coefficients, and the channel-aware mode selection matrix. To address the non-convex binary and cardinality constraints introduced by dynamic mode selection, we propose an efficient alternating optimization (AO) framework that employs penalty-based fractional programming (FP) and successive convex approximation (SCA) transformations. Simulation results demonstrate the potential of RDARS in enhancing the secrecy rate and show its superiority compared to existing reflection surface-based schemes.

\end{abstract}

\begin{IEEEkeywords} 
Physical-layer security, reconfigurable distributed antenna and reflecting surface (RDARS), reconfigurable intelligent surface (RIS), secure communication 
\end{IEEEkeywords}

\section{Introduction}
As a new wireless communication technology emerging in recent years, reconfigurable intelligent surfaces (RISs) play an important role in the future development of 6G \cite{10596064}. A RIS leverages a large number of programmable electromagnetic units to flexibly regulate incident electromagnetic waves and dynamically reshape the propagation environment, thereby improving coverage and transmission efficiency for wireless communications \cite{9475160}.

Due to its high potential for energy and spectrum efficiency, numerous research efforts have been conducted to explore the benefits of RIS-aided communication systems. However, the fully passive nature of RIS limits its reflection gains due to the ``multiplicative fading'' effect \cite{zhang2023active}. To overcome this fundamental physical limitation while retaining the advantages of RIS, \cite{ma2023reconfigurable} proposed a novel flexible architecture called reconfigurable distributed antenna and reflection surface (RDARS), which encompasses distributed antenna systems (DAS) and RIS as two special cases.
RDARS combines the flexibility of distributed antennas with the reconfigurability of passive reflecting surfaces, allowing each element to be dynamically programmed in two modes: reflection mode and connection mode. Elements in reflection mode function similarly to conventional passive RIS, while those in connection mode act as distributed antennas to directly transmit or receive signals.

Researchers have applied RDARS to improve communication capacity, transmission reliability, and indoor localization, as seen in \cite{ma2023reconfigurable, ma2024, wang2024, wang2024spawc, 10233300}. For example, the authors in \cite{ma2023reconfigurable} derived closed-form expressions for the ergodic achievable rate with both optimal and arbitrary RDARS configurations using maximum ratio combining (MRC), and provided experimental results for a prototype of the RDARS-aided system with 256 elements. The extension of RDARS to a massive multiple-input multiple-output (MIMO) uplink communication scenario with a two-time scale (TTS) transceiver design was explored in \cite{ma2024}. Additionally, \cite{wang2024} investigated the system transmission reliability for RDARS-aided uplink multi-user communication systems and proposed an inexact block coordinate descent (BCD)-based penalty dual decomposition (PDD) algorithm to minimize the sum minimum square error (MSE). However, the potential of RDARS in enhancing physical layer security remains unexplored.

In this letter, we leverage the flexible architecture of RDARS to enhance secrecy rate performance by jointly designing the active beamforming vectors, reflection coefficients, and the channel-aware mode selection matrix. Unlike existing reflection surface-aided secure communications \cite{cui2019secure, dong2022active, guan2020intelligent}, we consider the additional degree of freedom (DoF) introduced by channel-aware mode selection for RDARS. This leads to a more complex optimization problem involving binary and cardinality constraints.
To address these challenges, we employ penalty-based fractional programming (FP) and successive convex approximation (SCA) methods to handle the non-convex constraints, and propose an efficient alternating optimization (AO) framework to iteratively optimize each variable. Simulation results demonstrate the potential of RDARS in enhancing the secrecy rate, showing its superiority over existing reflection surface-based schemes.

\section{System Model and Problem Formulation}
    
\subsection{System Model}
\begin{figure}[h]
    \centering  
     \includegraphics[width=0.4\textwidth]{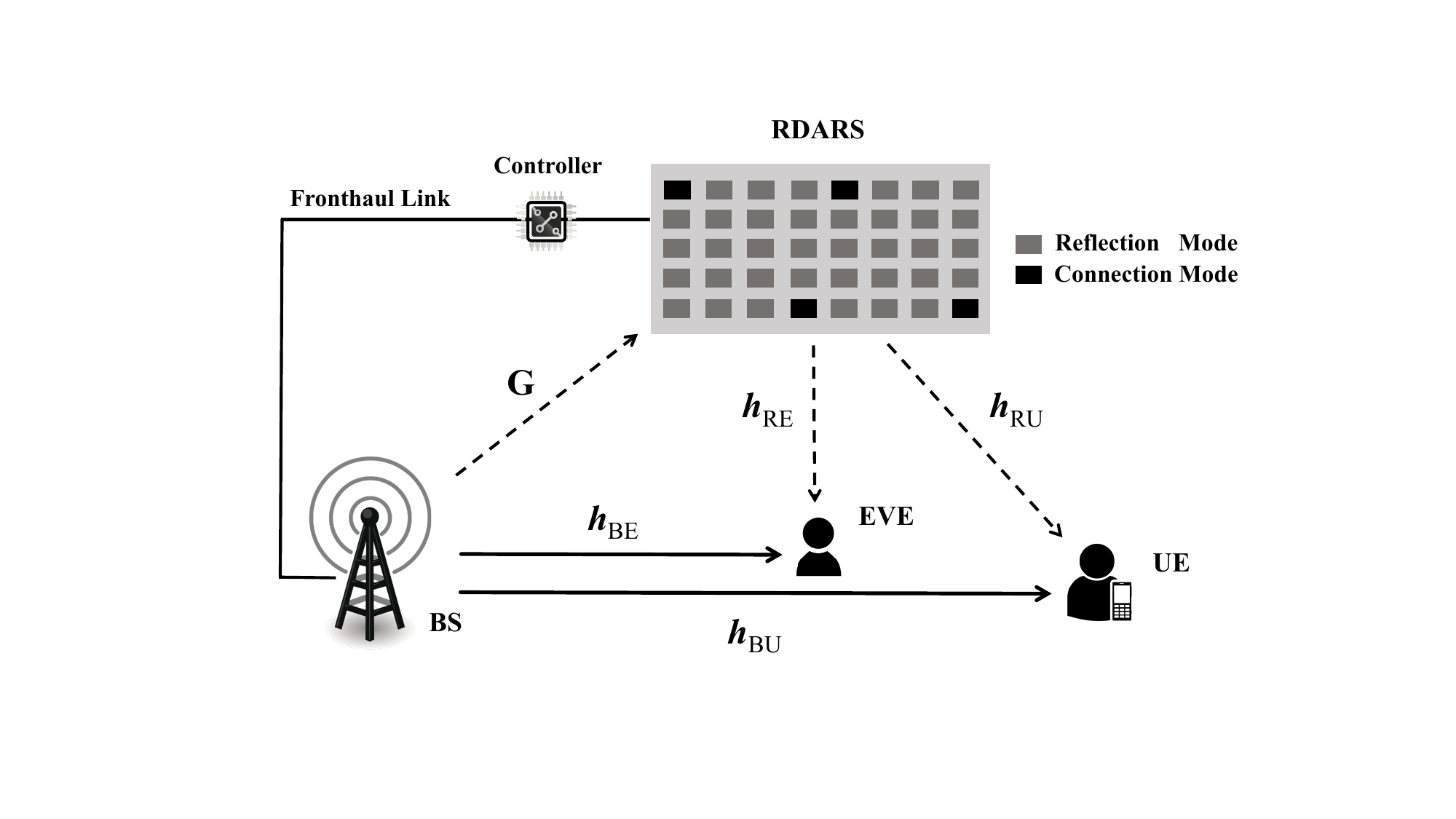} 
     \caption{A RDARS-aided secure transmission system.}
     \vspace{-12pt}
     \label{system_model} 
\end{figure}
We consider a RDARS-driven downlink secure communication system comprising one single-antenna legitimate user equipment (UE), one single-antenna eavesdropper (EVE), one  $N_t$-antenna base station (BS) and one RDARS equipped with $N$ elements, as illustrated in Fig.~\ref{system_model}.  
Each element in the RDARS can be dynamically programmed to operate in either reflection or connection mode via the controller \cite{ma2023reconfigurable}. Let $a$ denote the number of elements operating in connection mode, while the remaining $N-a$ elements operate in reflection mode. 
We define $\mathcal{S}$ and $\mathcal{N}$ as the index sets of the elements in connection mode and the total elements, respectively. Additionally, we introduce the binary diagonal mode selection matrix  ${\bf{S}}\in \mathbb{R}^{N \times N}$, where the diagonal values are set to 1 for elements in connection mode and 0 for those in reflection mode, i.e.,
\begin{align}
     [{\bf{S}}]_{n,n}=
     \left\{
        \begin{aligned}
        1, ~& \forall n \in \mathcal{S}, \\
        0, ~ & \forall n \notin \mathcal{S}.
        \end{aligned}
    \right.
\end{align}
The reflection coefficient matrix for RDARS is denoted as ${\bm{\Theta}}={\rm diag}({\bm{\theta}})$ where ${\bm{\theta}}=[{\theta}_{1}^H,{\theta}_2^H,\cdots,{\theta}_N^H]^H$.
The channel coefficients are defined as follows: ${\bf{G}}\in\mathbb{C}^{N_t \times N}$ represents the channel from the BS to the RDARS, ${\bm{h}}_{\rm BU}\in\mathbb{C}^{N_t \times 1}$ is the channel from the BS to the legitimate user, ${\bm{h}}_{\rm BE}\in\mathbb{C}^{N_t \times 1}$ is the channel from the BS to the eavesdropper, ${\bm{h}}_{\rm RU}\in\mathbb{C}^{N \times 1}$ is the channel from the RDARS to the legitimate user, and ${\bm{h}}_{\rm RE}\in\mathbb{C}^{N \times 1}$ is the channel from the RDARS to the eavesdropper.
To explore the performance limits of the considered RDARS-enhanced secure communication system, we assume that all channel state information (CSI) is perfectly known \footnote{The channel information can be obtained using methods proposed in \cite{jin2021channel}. The investigation of the impact of imperfect channel information on RDARS-aided secure transmission is left for future work.}.


Denote $s$ as the confidential message satisfying ${\mathbb{E}}[{|s|^2}]=1$. 
The received signals at the legitimate user and the eavesdropper can be respectively expressed as
\begin{subequations}\label{signal:label}
\begin{align}
    {y}_{u} & \!=\! ({\bm{h}}_{\rm BU}^H \!+\!  {\bm{h}}_{\rm RU}^H  ({\bf{I}}\!-\!{\bf{S}}){\bm{\Theta}} {\bf{G}}^H) {\bm{\omega}}_b s \!+ {\bm{h}}_{\rm RU}^H {\bf{S}}_a^H  {\bm{\omega}}_r s \!+\! n_u, \\
    {y}_{e} & \!=\! ({\bm{h}}_{\rm BE}^H \!+\!  {\bm{h}}_{\rm RE}^H  ({\bf{I}}\!-\!{\bf{S}}){\bm{\Theta}} {\bf{G}}^H){\bm{\omega}}_bs \!+{\bm{h}}_{\rm RE}^H {\bf{S}}_a^H{\bm{\omega}}_r s \!+\! n_e,
\end{align}
\end{subequations}
where ${\bm{\omega}}_b\in\mathbb{C}^{N_t \times 1}$ and  ${\bm{\omega}}_r\in\mathbb{C}^{a \times 1}$ denote the active beamforming vectors for the BS and the connected elements at the RDARS, respectively. 
$n_u$ and $n_e$ denote the received additive white Gaussian noise (AWGN) at the legitimate user and the eavesdropper with zero mean and variance of $\sigma_u^2$ and $\sigma_e^2$, respectively.
The fat matrix ${\bf{S}}_a \in \mathbb{C}^{a \times N}$ in \eqref{signal:label} is a submatrix of ${\bf{S}}$ consisting of non-zero rows of the mode selection matrix ${\bf{S}}$, satisfying ${\bf{S}}_a^H {\bf{S}}_a={\bf{S}}$ and ${\bf{S}}_a {\bf{S}}_a^H = {\bf{I}}_a$. Therefore, ${\bf{S}}_a{\bm{h}}_{\rm RU}$ and ${\bf{S}}_a{\bm{h}}_{\rm RE}$ represent the sub-channels associated with the connected elements at the RDARS.

Denote $ {\bm{h}}_{\rm ub} \triangleq {\bm{h}}_{\rm BU} +   {\bf{G}} {\bm{\Theta}}^H ({\bf{I}}-{\bf{S}}) {\bm{h}}_{\rm RU} $ and ${\bm{h}}_{\rm ur}\triangleq {\bf{S}}_a{\bm{h}}_{\rm RU}$ as the effective channels associated with the elements in the reflection and connection modes for the legitimate user, respectively. Similarly, $ {\bm{h}}_{\rm eb} \triangleq {\bm{h}}_{\rm BE} +   {\bf{G}} {\bm{\Theta}}^H ({\bf{I}}-{\bf{S}}) {\bm{h}}_{\rm RE} $ and ${\bm{h}}_{\rm er}\triangleq {\bf{S}}_a{\bm{h}}_{\rm RE}$. Then, the secrecy rate from the BS to the legitimate user in bits/second/Hertz (bps/Hz) can be expressed as \cite{k.a2010secure}
\begin{align}
    R_{s} = \left[ R_{u}-R_{e} \right]^{+},
\end{align}
where $\left[ z\right]^{+}\triangleq {\max (z,0)}$. 
$R_{u}$ and $R_{e}$ denote the achievable rates of the legitimate link and the eavesdropper link, respectively, defined as 
\begin{subequations} 
\begin{align}
    R_{u} & = \log_2 \left(  1 + \frac{|{\bm{h}}_{\rm ub}^H{\bm{\omega}}_b \!+\! {\bm{h}}_{\rm ur}^H{\bm{\omega}}_r |^2}{\sigma_u^2} \right), \label{R_u} \\ 
    R_{e} & = \log_2 \left(  1 + \frac{|{\bm{h}}_{\rm eb}^H{\bm{\omega}}_b \!+\! {\bm{h}}_{\rm er}^H{\bm{\omega}}_r |^2}{\sigma_e^2} \right). \label{R_e} 
\end{align}
\end{subequations}




\subsection{Problem Formulation}

In this letter, we aim to maximize the secrecy rate by jointly designing the active beamforming vectors $\{{\bm{\omega}}_b,{\bm{\omega}}_r\}$, the reflection coefficient matrix ${\bm{\Theta}}$ and the channel-aware mode selection matrix ${\bf{S}}$.
Accordingly, the secrecy rate maximization problem is formulated as \footnote{The operator $\left[ \cdot \right]^{+}$ in the objective function \eqref{P1_obj} is neglected without loss of optimality since the optimal value must be non-negative. This can be proven by contradiction: If $R_u-R_e<0$, we can increase its value to zero by setting ${\bm{\omega}}_b={\bm{\omega}}_r={\bf{0}}$ without violating the constraints.} 
 \begin{subequations}\label{Problem P1}
  \begin{align}
  (\text{P0}): \ \max_{{\bf{S}},{\bf{S}}_a,{\bm{\Theta}},{\bm{\omega}}_b,{\bm{\omega}}_r } \quad  \!\!& R_u-R_e   \label{P1_obj}  \\
 \!\!\mbox{s.t.} \qquad 
        & ||{\bm{\omega}}_b||^2+||{\bm{\omega}}_r||^2 \leq P ,\label{power_subcarrier} \\
       & | {\bf{\Theta}} |_{n,n}=1, ~\forall n \in \mathcal{N},
       \label{P1_RIS}  \\
       & [{\bf{S}}]_{n,n} \in \{0,1\}, \forall n \in \mathcal{N},  \label{P1_index} \\
       & [{\bf{S}}]_{i,j}=0, \forall i,j \in \mathcal{N}, i \neq j, \label{P1_index_diag} \\
       & ||{\rm diag}({\bf{S}})||_0=a, \label{P1_index_number}, \\ 
       &  {\bf{S}}_a^H{\bf{S}}_a={\bf{S}}, \label{mode equ}
       
  \end{align} 
 \end{subequations}
where \eqref{power_subcarrier} denotes the transmit power budget, and \eqref{P1_RIS} represents the unit-modulus constraints for the reflection coefficients at the RDARS. The binary constraints \eqref{P1_index} and \eqref{P1_index_diag} indicate the operational mode for each element at the RDARS, while the cardinality constraint \eqref{P1_index_number} limits the number of elements in connection mode to $a$.

Unfortunately, problem (P0) features a complex objective function and non-convex unit-modulus constraints, making its optimal solution challenging to determine. Furthermore, the binary and cardinality constraints complicate the problem further, rendering (P0) a mixed-integer optimization problem, which is intractable and NP-hard. To address these challenges, we propose an efficient AO algorithm in the following section.

\section{Joint Beamforming and Mode Selection}
In this section, we first employ penalty-based FP and SCA methods to address the binary and cardinality constraints in problem (P0) and then propose an efficient AO framework to optimize each variable iteratively.

\subsection{Mode Selection} 
For mode selection, it involves the effective channels $\{ {\bm{h}}_{\rm ub},{\bm{h}}_{\rm ur}, {\bm{h}}_{\rm eb},{\bm{h}}_{\rm er} \}$, which are influenced by ${\bf{S}}$ and ${\bf{S}}_a$. 
To address the coupling effect between these two mode selection matrices, we propose a penalty-based method to solve ${\bf{S}}$ and ${\bf{S}}_a$ iteratively. 
By introducing a penalty factor to the equality constraint in \eqref{mode equ}, we reformulate the mode selection subproblem as follows:
 \begin{subequations}\label{Problem mode}
  \begin{align}
  (\text{P1}): \ \max_{{\bf{S}},{\bf{S}}_a} \quad  \!\!& \frac{1 + \frac{1}{\sigma_u^2} |{\bm{h}}_{\rm ub}^H{\bm{\omega}}_b \!+\! {\bm{h}}_{\rm ur}^H{\bm{\omega}}_r |^2 }{ 1 + \frac{1}{\sigma_e^2} |{\bm{h}}_{\rm eb}^H{\bm{\omega}}_b \!+\! {\bm{h}}_{\rm er}^H{\bm{\omega}}_r |^2 } - \frac{\|\mathbf{S}\!-\!\mathbf{S}_{a}^H \mathbf{S}_{a}\|_F^{2}}{2\rho}    \label{P1_obj}  \\
 \!\!\mbox{s.t.} \quad 
        & \eqref{P1_index}, \eqref{P1_index_diag}, \eqref{P1_index_number}.     
  \end{align} 
 \end{subequations}
where $\rho>0$ is the penalty factor.  
In the sequel, we propose using the FP and SCA methods to solve this problem.
\subsubsection{Optimizing ${\bf{S}}$ with given ${\bf{S}}_a$}
We first decouple ${\bf{S}}$ from $ {\bm{h}}_{\rm ub}$ and $ {\bm{h}}_{\rm eb}$ in the objective function. Letting ${\bm{s}}={\rm diag}({\bf{S}} )$, the following equalities hold:
\begin{subequations}
\begin{align}
 {\bm{h}}_{\rm ub}^H{\bm{\omega}}_b \!+\! {\bm{h}}_{\rm ur}^H{\bm{\omega}}_r &= -{\bm{s}}^H {\bm{p}}_u + d_u, \\
 {\bm{h}}_{\rm eb}^H{\bm{\omega}}_b \!+\! {\bm{h}}_{\rm er}^H{\bm{\omega}}_r & = -{\bm{s}}^H {\bm{p}}_e + d_e,
\end{align}
\end{subequations}
with ${\bm{p}}_u={\rm diag}({\bm{h}}_{\rm RU}^H) {\bm{\Theta}} {\bf{G}}^H {\bm{\omega}}_b$ and $d_u = ({\bm{h}}_{\rm BU}^H + {\bm{h}}_{\rm RU}^H{\bm{\Theta}} {\bf{G}}^H) {\bm{\omega}}_b + {\bm{h}}_{\rm ur}^H{\bm{\omega}}_r$. Similarly, ${\bm{p}}_e={\rm diag}({\bm{h}}_{\rm RE}^H) {\bm{\Theta}} {\bf{G}}^H {\bm{\omega}}_b$ and $d_e = ({\bm{h}}_{\rm BE}^H + {\bm{h}}_{\rm RE}^H{\bm{\Theta}} {\bf{G}}^H) {\bm{\omega}}_b + {\bm{h}}_{\rm er}^H{\bm{\omega}}_r$.
Defining ${\bm{\tilde {s}}}\triangleq [{\bm{{s}}}^T,1]^T$, the subproblem for optimizing ${\bf{S}}$ is formulated as the following sum-of-ratio problem
\begin{subequations}
  \begin{align}
  (\text{P2}):\ \max_{{\bm{\tilde {s}}} } \quad  \!\!&    \frac{{\bm{\tilde{s}}}^H {{\mathbf{P}}}_{u}  {\bm{\tilde{s}}}}{{{\bm{\tilde{s}}}^H {{\mathbf{P}}}_{e}  {\bm{\tilde{s}}}}}-\frac{ {\bm{\tilde {s}}}^H {\bm{p}} }{2\rho} \label{P4_obj} \\
 \!\!\mbox{s.t.} \quad 
       & [{\bm{\tilde {s}}}]_{n} \in \{0,1\}, \forall n \in \mathcal{N}, \label{p4_1}\\
       & [{\bm{\tilde {s}}}]_{N+1} = 1, \label{p4_2} \\
       & {\bm{\tilde {s}}}^H{\bm{\tilde {s}}}=a+1, \label{p4_3}
  \end{align} 
 \end{subequations}
with ${\mathbf{P}}_{u}= \left[ {\sigma_u^{-2}}{\bm{p}}_u{\bm{p}}_u^H, -{\sigma_u^{-2}}d_u^*{\bm{p}}_u; -{\sigma_u^{-2}}d_u{\bm{p}}_u^H, {\sigma_u^{-2}}|d_u|^2 +1  \right]$, ${\mathbf{P}}_{e}= \left[ {\sigma_e^{-2}}{\bm{p}}_e{\bm{p}}_e^H, -{\sigma_e^{-2}}d_e^*{\bm{p}}_e; -{\sigma_e^{-2}}d_e{\bm{p}}_e^H, {\sigma_e^{-2}}|d_e|^2 +1  \right]$ and 
${\bm{p}}= [\operatorname{diag}(\mathbf{I}-2\mathbf{S}_{a}^H \mathbf{S}_{a});{\rm tr}(\mathbf{S}_{a}^H {\mathbf{S}}_{a})]$.
The multiple-ratio fractional problem is generally NP-hard. Therefore, we propose to equivalently reformulate the sum-of-ratio problem using the quadratic transformation described in the proposition below.
\begin{proposition}
    Problem (P2) is equivalent to  
    \begin{subequations}\label{multiple ratio}
  \begin{align}
   \max_{{\bm{\tilde {s}}},{\bm{\beta}} } \quad  \!\!&   2\Re\{  {\bm{\beta}}^H {\bm{\Lambda}}_{u}^{\frac{1}{2}}{\bf{U}}_{u}^H  {\bm{\tilde{s}}}\}-
  {{{\bm{\tilde{s}}}^H {{\mathbf{P}}}_{e}  {\bm{\tilde{s}}}}} {\bm{\beta}}^H{\bm{\beta}} -\frac{ {\bm{\tilde {s}}}^H {\bf{p}} }{2\rho} \label{multiple ratio obj} \\
 \!\!\mbox{s.t.} \quad 
       & \eqref{p4_1},\eqref{p4_2},\eqref{p4_3},
  \end{align} 
 \end{subequations}
where ${\bm{\beta}}\in\mathbb{C}^{N+1}$ refer to the auxiliary variable. ${\bf{U}}_{u}$ and ${\bm{\Lambda}}_{u}$ are obtained by the  eigen-value decomposition (EVD) w.r.t. ${\mathbf{P}}_{u}$, i.e., ${\mathbf{P}}_{u}={\bf{U}}_{u}{\bm{\Lambda}}_{u}{\bf{U}}_{u}^H$.
\end{proposition}
\begin{proof}
    It can be easily inferred that the optimal solution of \eqref{multiple ratio} is ${\bm{\beta}}^{\rm opt}=({{{\bm{\tilde{s}}}^H {{\mathbf{P}}}_{e}  {\bm{\tilde{s}}}}})^{-1}{\bm{\Lambda}}_{u}^{\frac{1}{2}}{\bf{U}}_{u}^H  {\bm{\tilde{s}}}$ and the optimal objective of \eqref{multiple ratio obj} equals to \eqref{P4_obj} exactly. Thus, the equivalence to \eqref{multiple ratio} is therefore established \cite{8314727}.
\end{proof}
Since the non-convex quadratic problem in \eqref{multiple ratio} with binary constraints is still difficult to solve and has no closed-form solution, a convex surrogate function of the quadratic term ${\bm{\tilde{S}}}^H {{\mathbf{P}}}_{e}  {\bm{\tilde{s}}}$ is therefore found based on the second-order Taylor expansion at $t$-th iteration, i.e., 
\begin{align}
{\bm{\tilde{s}}}^H {{\mathbf{P}}}_{e}  {\bm{\tilde{s}}} \leq & 2\lambda_{\rm max}(a+1)-{\bm{\tilde{s}}}_{t}^H {{\mathbf{P}}}_{e} {\bm{\tilde{s}}}_{t} + \notag\\
& 2\Re\{  {\bm{\tilde{s}}}_{t}^H ( {{\mathbf{P}}}_{e}-\lambda_{\rm max}{\bf{I}}){\bm{\tilde{s}}} \},
\end{align}
where $\lambda_{\rm max}$ represents the maximum eigenvalue of ${{\mathbf{P}}}_{e}$.
The problem for optimizing ${\bm{\tilde s}}$ is finally formulated as 
  \begin{align}
   \max_{{\bm{\tilde {s}}} } \quad  \!\!   \Re\{ {\bm{\gamma}}^H {\bm{\tilde {s}}} \}
~\quad \mbox{s.t.}  ~
  \eqref{p4_1},\eqref{p4_2},\eqref{p4_3},
  \end{align} 
with ${\bm{\gamma}}={\bf{U}}_{u} {\bm{\Lambda}}_{u}^{\frac{1}{2}} {\bm{\beta}}+ ( \lambda_{\rm max}{\bf{I}}-{{\mathbf{P}}}_{e}){\bm{\tilde{s}}}_{t}{\bm{\beta}}^H {\bm{\beta}} - \frac{1}{4\rho}{\bf{p}}$. Denoting the index set corresponding to the first $N$ largest value of $\Re\{{\bm{\gamma}}\}$  as ${\mathcal{A}}^{\rm opt}$, the optimal ${\bm{ {s}}}^{\rm opt}$ can be obtained as
 \begin{align}\label{s update}
&[{\bm{{s}}}^{\rm opt}]_n=\begin{cases}
1,& \forall n \in {\mathcal{A}}^{\rm opt},\\
0, &\forall n \notin {\mathcal{A}}^{\rm opt}.
\end{cases}
\end{align}

\subsubsection{Optimizing ${\bf{S}}_a$ with given ${\bf{S}}$}
To decouple ${\bf{S}}_a$ from the objective function in \eqref{Problem mode},
we define ${\bm{z}} \triangleq {\rm vec}({\bm{S}}_a^T)$ and ${\bm{\tilde z}}\triangleq[{\bm{z}}^T,1]^T$ with ${\bm{z}}=\left[ {\bm{z}}_{1}^{T}, {\bm{z}}_{2}^{T}, \dots, {\bm{z}}_{a}^{T} \right]^T$ and ${\bm{z}}_i \neq {\bm{z}}_j, \forall i \neq j$.
${\bm{z}}_i \in \mathcal{S}_{a}\triangleq \{ {\bm{z}}\in \mathbb{R}^{N} | {\bf{1}}^T{\bm{z}}=1, [{\bm{z}}]_n \in \{0,1\},\forall n \}$ denotes that each ${\bm{z}}_i$ belongs to a special ordered set of type 1 (SOS1).   
Thus, the subproblem for optimizing ${\bf{S}}_a$ is formulated as
\begin{subequations}
  \begin{align}
   (\text{P3}):\ \max_{ {\bm{\tilde z}} } \quad  \!\!&    \frac{ {\bm{\tilde z}}^H {\mathbf{H}}_{u} {\bm{\tilde z}} } { {\bm{\tilde z}}^H {\mathbf{H}}_{e}  {\bm{\tilde z}} } - {\bm{\tilde z}}^H{\mathbf{M}} {\bm{\tilde z}}  \\
 \!\!\mbox{s.t.} \quad  
        & [{\bm{z}}_i]_n \in \{0,1\}, \forall n\in \mathcal{N}, i=\{1,...,a\}, \label{p3_1}\\
        & [{\bm{z}}]_{aN+1}=1, \label{p3_2}\\
        & {\bm{z}}_i^T{\bm{z}}_i=1, \forall  i=\{1,...,a\}, \label{p3_3}
  \end{align} 
 \end{subequations}
where ${\mathbf{H}}_{u}=[{\sigma_u^{-2}}({\bm{\omega}}_r\otimes{\bm{h}}_{\rm RU})({\bm{\omega}}_r^H\otimes{\bm{h}}_{\rm RU}^H),{\sigma_u^{-2}}({\bm{\omega}}_r{\bm{\omega}}_b^H{\bm{h}}_{\rm ub})\otimes{\bm{h}}_{\rm RU};{\sigma_u^{-2}}({\bm{\omega}}_r^H{\bm{h}}_{\rm ub}^H{\bm{\omega}}_b)\otimes{\bm{h}}_{\rm RU}^H,{\sigma_u^{-2}}|{\bm{h}}_{\rm ub}^H{\bm{\omega}}_b|^2+1]$, ${\mathbf{H}}_{e}=[{\sigma_e^{-2}}({\bm{\omega}}_r\otimes{\bm{h}}_{\rm RE})({\bm{\omega}}_r^H\otimes{\bm{h}}_{\rm RE}^H),{\sigma_e^{-2}}({\bm{\omega}}_r{\bm{\omega}}_b^H{\bm{h}}_{\rm eb})\otimes{\bm{h}}_{\rm RE};{\sigma_e^{-2}}({\bm{\omega}}_r^H{\bm{h}}_{\rm eb}^H{\bm{\omega}}_b)\otimes{\bm{h}}_{\rm RE}^H,{\sigma_e^{-2}}|{\bm{h}}_{\rm eb}^H{\bm{\omega}}_b|^2+1]$ and  ${\bf{M}}=[-\rho^{-1}({\bf{I}}_a\otimes{\bf{S}}),0;0,0]$.
 The same methodology for solving problem (P2) can be leveraged to address problem (P3).
Briefly, the problem (P3) can be transformed into the following problem with FP and SCA approaches:
  \begin{align}\label{s_a}
   \max_{{\bm{\tilde {z}}} } \quad  \!\!   \Re\{ {\bm{\epsilon}}^H {\bm{\tilde {z}}} \}
~\quad \mbox{s.t.}  ~
  \eqref{p3_1},\eqref{p3_2},\eqref{p3_3},
  \end{align} 
 where ${\bm{\epsilon}}={\bf{U}}_e{\bf{\Lambda}}_e^{\frac{1}{2}}{\bm{\eta}}+({\mu_{\rm max}}{\bf{I}}-{\bf{H}}_e{\bm{\eta}}^H{\bm{\eta}}-{\bf{M}}){\bm{\tilde {z}}}_t$ with  ${\bf{U}}_e$ and ${\bf{\Lambda}}_e$ are obtained by EVD w.r.t ${\bf{H}}_u$. ${\mu_{\rm max}}$ represents the maximum eigenvalue of ${\bf{H}}_e{\bm{\eta}}^H{\bm{\eta}}+{\bf{M}}$ and the auxiliary variable ${\bm{\eta}}$ is updated with ${\bm{\eta}}^{\rm opt}=({\bm{\tilde z}}^H {\mathbf{H}}_{e} {\bm{\tilde z}})^{-1}{\bf{\Lambda}}_e^{\frac{1}{2}}{\bf{U}}_e^H{\bm{\tilde {z}}} $.
 Regarding problem \eqref{s_a}, the optimal ${\bm{ {z}}}^{\rm opt}$ can be obtained via the one-dimensional research method.

\subsection{Active Beamforming}
Denoting ${\mathbf{R}}_{u}={\sigma_u^{-2}} \left[ {\bm{h}}_{\rm ub}^H, {\bm{h}}_{\rm ur}^H \right]^H \left[ {\bm{h}}_{\rm ub}^H, {\bm{h}}_{\rm ur}^H \right]$ and $ {\mathbf{R}}_e={\sigma_e^{-2}} \left[ {\bm{h}}_{\rm eb}^H, {\bm{h}}_{\rm er}^H \right]^H \left[ {\bm{h}}_{\rm eb}^H, {\bm{h}}_{\rm er}^H \right]$, the subproblem for active beamforming can be formulated as the following generalized Rayleigh quotient:
  \begin{align}
   (\text{P4}): \ \max_{{\bm{\omega}}} \quad  \!\! \frac{{\bm{\omega}}^H {\mathbf{R}}_u {\bm{\omega}}+1}{{\bm{\omega}}^H {\mathbf{R}}_e {\bm{\omega}}+1} 
 \quad \mbox{s.t.} ~ 
        ||{\bm{\omega}}||^2 \leq P,
  \end{align} 
 where ${\bm{\omega}}=\left[{\bm{\omega}}_{b}^H, {\bm{\omega}}_{r}^H \right]^H$ represents the total active beamforming vector.
 Thus, the optimal solution ${\bm{\omega}}^{\rm opt}$ can be readily derived as 
%
\begin{align}\label{w update}
&{\bm{\omega}}^{\rm opt}= \sqrt{P}{\bf{u}}_{\rm max},
\end{align}
where ${\bf{u}}_{\rm max}$ is the normalized eigenvector corresponding to the
largest eigenvalue of the matrix $\left({\mathbf{R}}_e+\frac{1}{\sqrt{P}}\bf{I}\right)^{-1}\left({\mathbf{R}}_u+\frac{1}{\sqrt{P}}\bf{I}\right)$.\\

\begin{algorithm}[t] 
    \caption{Proposed Algorithm for Joint Beamforming and Mode Selection}     
     \label{xx}       
    \begin{algorithmic}[1] 
    \Require Set ${\bf{S}}^0$, ${\bf{\Theta}}^0$, ${\bm{\omega}}^0$ and $\rho^0$    
\Repeat
   \Repeat
   \State Update ${\bf{S}}$ via \eqref{s update}.
   \State Update ${\bf{S}}_a$ by solving problem \eqref{s_a}.
   \State Update penalty factor by $\rho=\alpha \rho$.
   \Until $\|\mathbf{S}\!-\!\mathbf{S}_{a}^H \mathbf{S}_{a}\|_F\leq \epsilon$.
   \State Update ${\bm{\omega}}$ by ${\bm{\omega}}^{\rm opt}= \sqrt{P}{\bf{u}}_{\rm max}$ in \eqref{w update}.
   \State Update ${\bf{\Theta}}$ by solving problem \eqref{theta problem} with Charnes-Cooper transformation and SDP. 
   
   \Until {the convergence is satisfied.}
   \Ensure  ${\bf{\Theta}}^{\star}$, ${\bm{\omega}}^{\star}$, ${\bf{S}}^{\star}$, ${\bf{S}}_a^{\star}$   
   \end{algorithmic} 
\end{algorithm} 

\vspace{-6pt}
\subsection{Reflection Coefficients}
Since the objective function and the unit-modulus constraints have a similar structure as the conventional RIS-aided ones, the reflection coefficient optimization algorithm derived in \cite{cui2019secure} can be applied to solve for the reflection coefficients.
In brief, we have  
\begin{subequations}
\begin{align}
 {\bm{h}}_{\rm ub}^H{\bm{\omega}}_b \!+\! {\bm{h}}_{\rm ur}^H{\bm{\omega}}_r &= {\bm{ \theta}}^H {\bf{q}}_u + c_u, \\
 {\bm{h}}_{\rm eb}^H{\bm{\omega}}_b \!+\! {\bm{h}}_{\rm er}^H{\bm{\omega}}_r & = {\bm{ \theta}}^H {\bf{q}}_e + c_e,
\end{align}
\end{subequations}
with ${\bf{q}}_u={\rm diag}({\bm{h}}_{\rm RU}^H)({\bf{I}}-{\bf{S}}) {\bf{G}}^H {\bm{\omega}}_b$ and $c_u = {\bm{h}}_{\rm BU}^H {\bm{\omega}}_b + {\bm{h}}_{\rm ur}^H{\bm{\omega}}_r$. Similarly, ${\bf{q}}_e={\rm diag}({\bm{h}}_{\rm RE}^H)({\bf{I}}-{\bf{S}}) {\bf{G}}^H {\bm{\omega}}_b$ and $c_e = {\bm{h}}_{\rm BE}^H {\bm{\omega}}_b + {\bm{h}}_{\rm er}^H{\bm{\omega}}_r$. By defining an auxilliary variable ${\bm{\tilde \theta}}=({\bm{ \theta}}^T,1)^T$, the reflection coefficient optimization problem can be recast to 
\begin{subequations}\label{theta problem}
  \begin{align}
  (\text{P5}):\ \max_{ {\bm{\tilde \theta}} } \quad  \!\!& \frac{{\bm{\tilde \theta}}^H{\mathbf{Q}}_{u}{\bm{\tilde \theta}}+{\tilde c}_{u}+1}{{\bm{\tilde \theta}}^H{\mathbf{Q}}_{e}{\bm{\tilde \theta}}+{\tilde c}_{e}+1}     \\
 \!\!\mbox{s.t.} \quad 
       &  {\bm{\tilde \theta}}^H{\mathbf{Y}}_{n}{\bm{\tilde \theta}}=1,\forall{n}.\label{theta constrant}
  \end{align} 
 \end{subequations}
with ${\mathbf{Q}}_{u}={\sigma_u^{-2}} \left[ {\bf{q}}_u{\bf{q}}_u^H, c_u^*{\bf{q}}_u; c_u{\bf{q}}_u^H,0  \right]$, ${\tilde c}_{u}={\sigma_u^{-2}}|c_u|^2$,
 ${\mathbf{Q}}_{e}={\sigma_e^{-2}} \left[ {\bf{q}}_e{\bf{q}}_e^H, c_e^*{\bf{q}}_e; c_e{\bf{q}}_e^H,0  \right]$ and ${\tilde c}_{e}={\sigma_e^{-2}}|c_e|^2$.
The $(i,j)$-th element of ${\mathbf{Y}}_{n}$, denoted by $[{\mathbf{Y}}_{n}]_{i,j}$, satisfies
 \begin{align}
       &[{\mathbf{Y}}_{n}]_{i,j}=\begin{cases}
1,& \text{$i=j=n$},\\
0, &\text{otherwise}.
\end{cases}
\end{align}

Fortunately, the above non-convex fractional quadratic optimization problem (P5) has been significantly studied. Specifically, the Charnes-Cooper transformation can be applied to convert problem (P5) into an equivalent non-fractional convex semidefinite programming (SDP) problem, which can be optimally solved using the interior-point method. For simplicity, we will omit the details here.

\vspace{-6pt}
\subsection{Overall Algorithm}
The overall procedure for solving problem (P0) is presented in {\bf{Algorithm 1}}. 
%
In the respective optimization process for the variables $\{ {\bf{S}}, {\bf{S}}_a, {\bf{\Theta}}, {\bm{\omega}} \}$, the objective value $R_s$ keeps non-decreasing, i.e., $R_s({\bf{\Theta}}^1, {\bm{\omega}}^1, {\bf{S}}^1, {\bf{S}}_a^1)\leq R_s({\bf{\Theta}}^2, {\bm{\omega}}^2, {\bf{S}}^2, {\bf{S}}_a^2)\leq \cdots \leq R_s({\bf{\Theta}}^i, {\bm{\omega}}^i, {\bf{S}}^i, {\bf{S}}_a^i)$, with ${\bf{\Theta}}^i, {\bm{\omega}}^i, {\bf{S}}^i, {\bf{S}}_a^i$ denoting the optimized ${\bf{\Theta}}, {\bm{\omega}}, {\bf{S}}, {\bf{S}}_a$ in $i$-th iteration.
Furthermore, the objective is upper-bounded by the constraints $\eqref{power_subcarrier}-\eqref{mode equ}$. As a result, the proposed algorithm is guaranteed to converge to a stationary point. 
 The complexity of {\bf{Algorithm 1}} primarily lies in steps 2 to 8. 
For steps 2 to 6, the complexity is calculated as ${\mathcal{O}}(I_{p}(N+1)^3)$, where $I_{p}$ denotes the number of iterations for the penalty-based iteration.
In contrast, the computational complexities for steps 7 and 8 are ${\mathcal{O}}((N_t+a)^3)$ and ${\mathcal{O}}((N+1)^{3.5})$, respectively. 
 Assuming the total number of iterations in the AO process is $I_{AO}$, the overall computational complexity of {\bf{Algorithm 1}} can be expressed 
as ${\mathcal{O}}\left(I_{AO}((N_t+a)^3+(N+1)^{3.5})\right)$.

\vspace{-12pt}
\section{Simulation}

\begin{figure}[t]
    \centering  
     \includegraphics[width=0.35\textwidth]{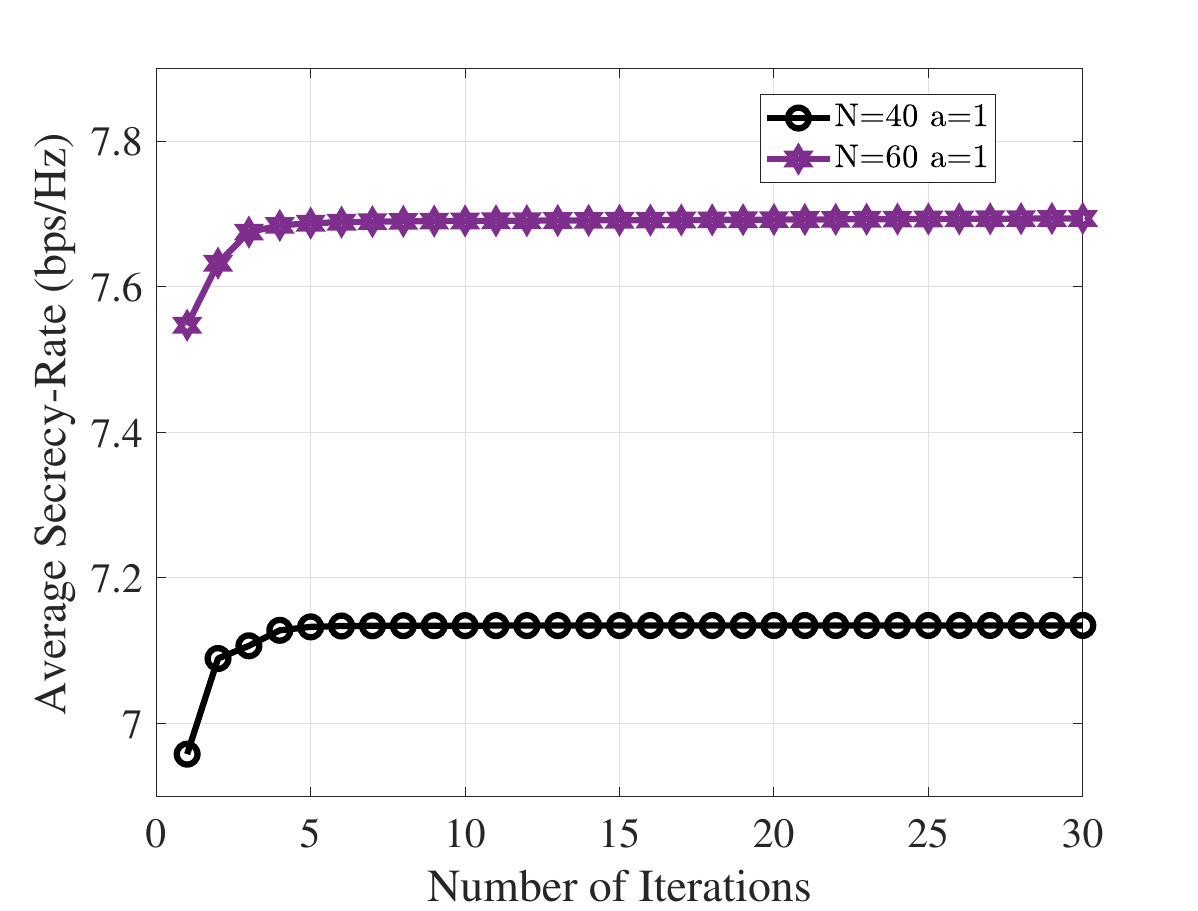} 
     \caption{Convergence behaviors of the proposed algorithm.}
     \vspace{-12pt}
     \label{conv} 
\end{figure}
In this section, we present numerical simulations to evaluate the performance of our proposed algorithm for joint beamforming design and channel-aware mode selection in RDARS-driven secure transmissions.
For simplicity, the BS and the RDARS are located at $(0, 0)$m and $(10, 5)$m, while the locations of UE and EVE are set as $(45, 0)$m and $(40, 0)$m, respectively. 
The small-scale fading in all-related channels is assumed to be Rician fading with a Rician factor of $1$ \cite{1237126}. 
The large-scale path loss is modeled as ${\rm PL}=\sqrt{\beta_0(d_0/d)^{\alpha_0}}$ where $\beta_0\!=\!-30$ dB denotes the path loss at the reference distance $d_0\!=\!1$m.
The path loss exponents are set as $\alpha_{\rm BU}\!=\!\alpha_{\rm BE}\!=\!3$ and  $\alpha_{\rm RU}\!=\!\alpha_{\rm RE}\!=\!\alpha_{\rm G}\!=\!2$. 
Unless otherwise stated, we set $N_t\!=\!5$, $\sigma_u^2\!=\!\sigma_e^2\!=\!\sigma_a^2\!=\!-90$ dBm,  $\epsilon\!=\!10^{-3}$ and $\rho^0\!=\!10^6$. All simulation results are obtained by averaging over 500 channel realizations. 

Several existing schemes are adopted as benchmarks, including conventional ``passive RIS'', ``active RIS'' and ``DAS''. ``RDARS Opt'' and ``RDARS w/o MS'' denote the RDARS-driven secure transmissions with the proposed penalty-based AO optimization and without mode selection, respectively.
For a fair comparison, we assume the number of BS antennas at ``passive RIS'' and ``active RIS'' is equal to the sum of BS antennas and connected elements at `RDARS Opt'' and ``RDARS w/o MS'', i.e., $N_t^{\rm PRIS/ARIS}\!=\!N_t^{\rm RDARS}\!+\!a$.

Fig.~\ref{conv} shows the convergence behavior of the proposed algorithm with varying numbers of RDARS elements. The simulation results demonstrate that the penalty-based AO approach converges quickly, typically within 10 iterations, across different configurations of RDARS elements.

In Fig.~\ref{fig_E1}, we compare the secrecy rate performance of different schemes under varying system parameters. 
Fig.~\ref{fig_E1} (a) shows the rate performance versus maximum transmit power $P$ when $a\!=\!2$ and $N\!=\!100$.
The RDARS-driven system achieves the highest secrecy rate among the schemes, demonstrating the appealing potential of its structure. Even with just one element in connection mode, the RDARS-driven system  significantly outperforms both active and passive RISs-assisted ones. Furthermore, ``RDARS w/o MS'' shows a performance improvement compared to DAS, which can be attributed to the reflection gain provided by the elements in the reflection modes. The performance gap between ``RDARS w/o MS'' and ``RDARS Opt'' highlights the selection gain achieved through the channel-aware mode selection of our proposed algorithm.
Fig.~\ref{fig_E1} (b) illustrates the rate performance versus the number of elements in RDARS when $P\!=\!15$ dBm and $a\!=\!2$.
 It can be observed that the secrecy rate increases as the total number of elements increases, except for ``DAS'', due to the reflection gains. Similarly, RDARS-driven schemes continue to perform the best even with a small number of connected elements.
Fig.~\ref{fig_E1} (c) presents the secrecy rate performance versus the number of
elements in the connection mode when $P\!=\!15$ dBm and $N\!=\!100$.  
The performance of ``passive RIS" and ``active RIS'' increases as the number of elements in the connection mode increases, as the number of BS antennas increases with a fair comparison.
Both RDARS-driven schemes and ``DAS" show significant performance enhancement, demonstrating the effectiveness of incorporating DAS to mitigate the effects of multiplicative fading. On the other hand, the performance gap between ``RDARS Opt'' and ``DAS'' decreases as $a$ increases. This can be attributed to the fact that the DoF for channel-aware mode placement of connected elements decreases and the distribution gain gradually dominates the reflection gain.


\begin{figure*}[htbp]
    \centering
    \subfigure[$P$] {\includegraphics[width=.33\textwidth]{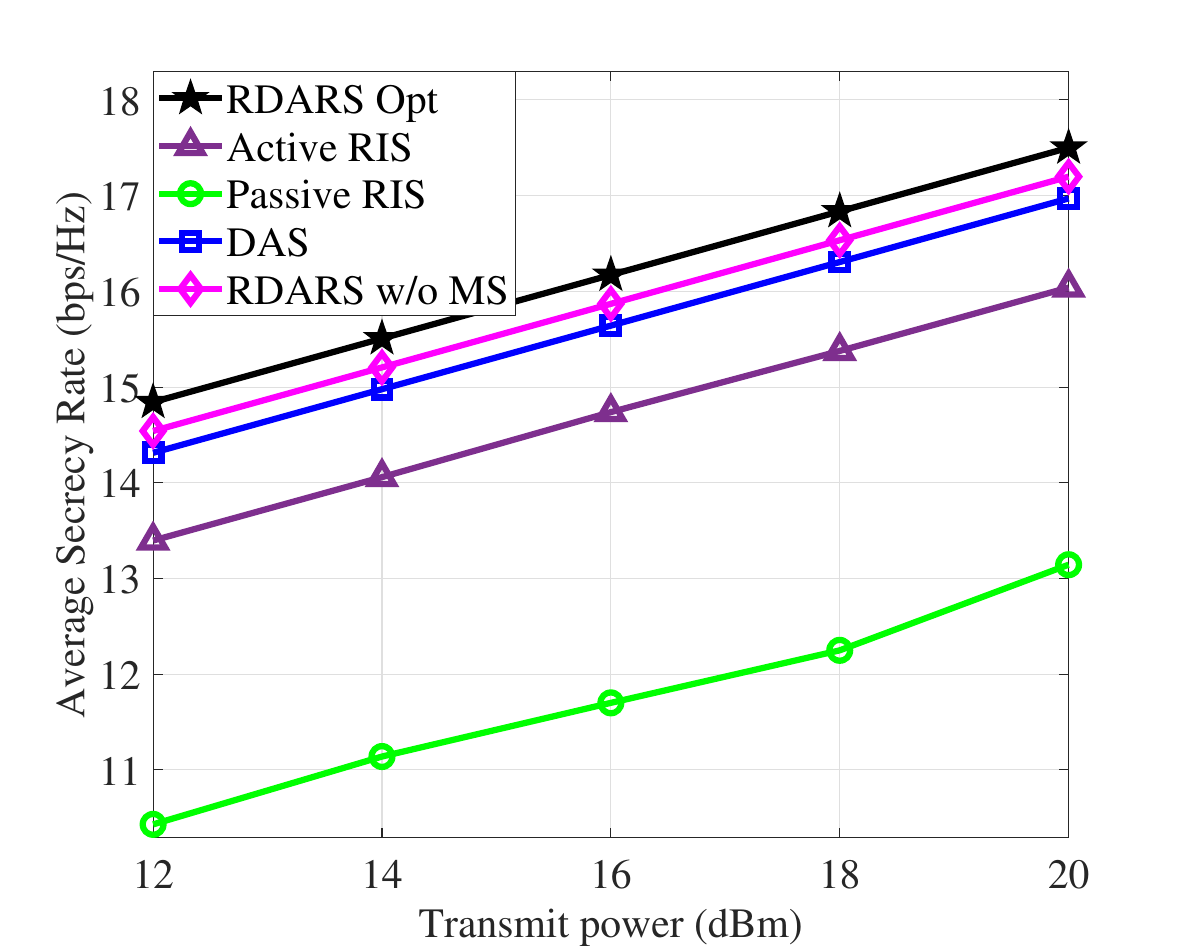}}
    \hspace{-6pt}
    \subfigure[$N$] {\includegraphics[width=.33\textwidth]{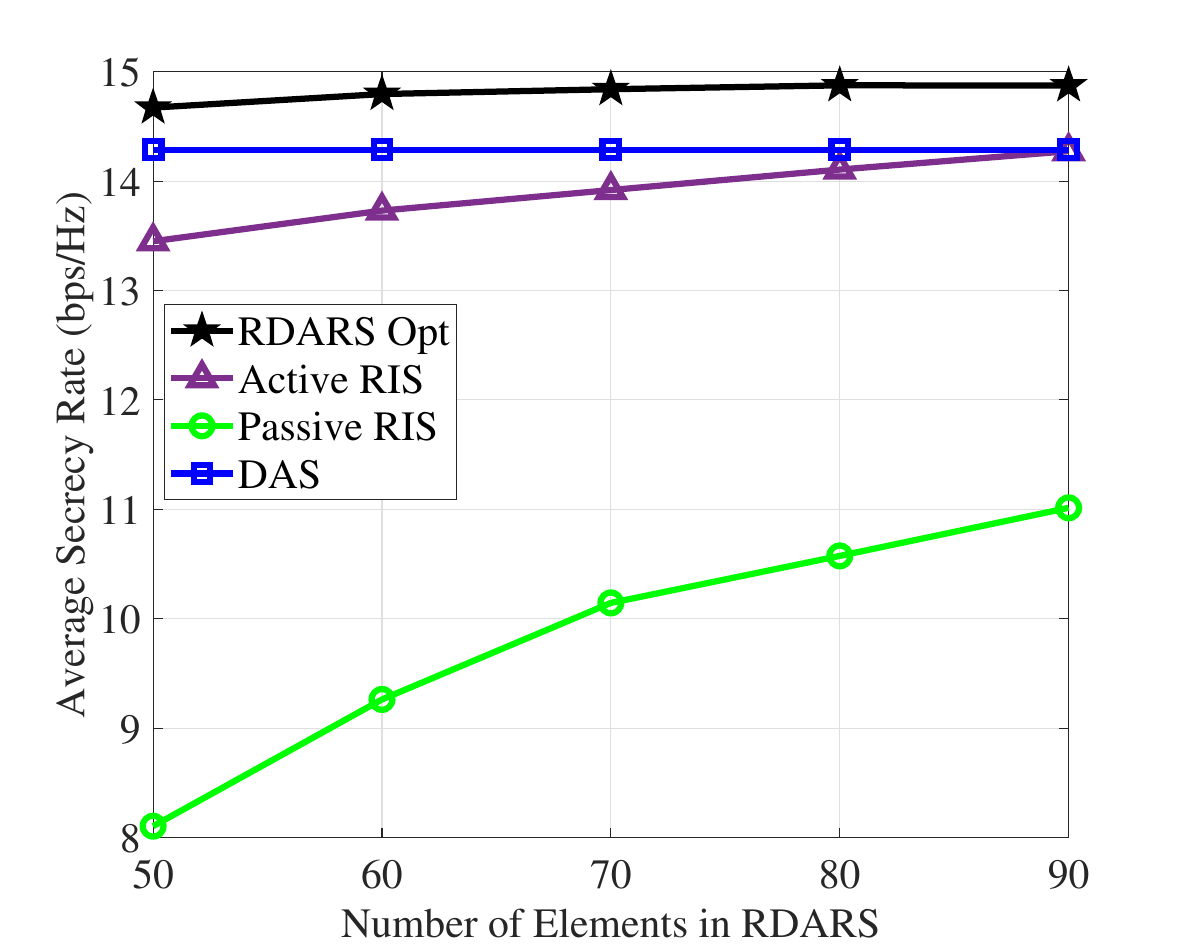}}
    \hspace{-6pt}
    \subfigure[$a$] {\includegraphics[width=.33\textwidth]{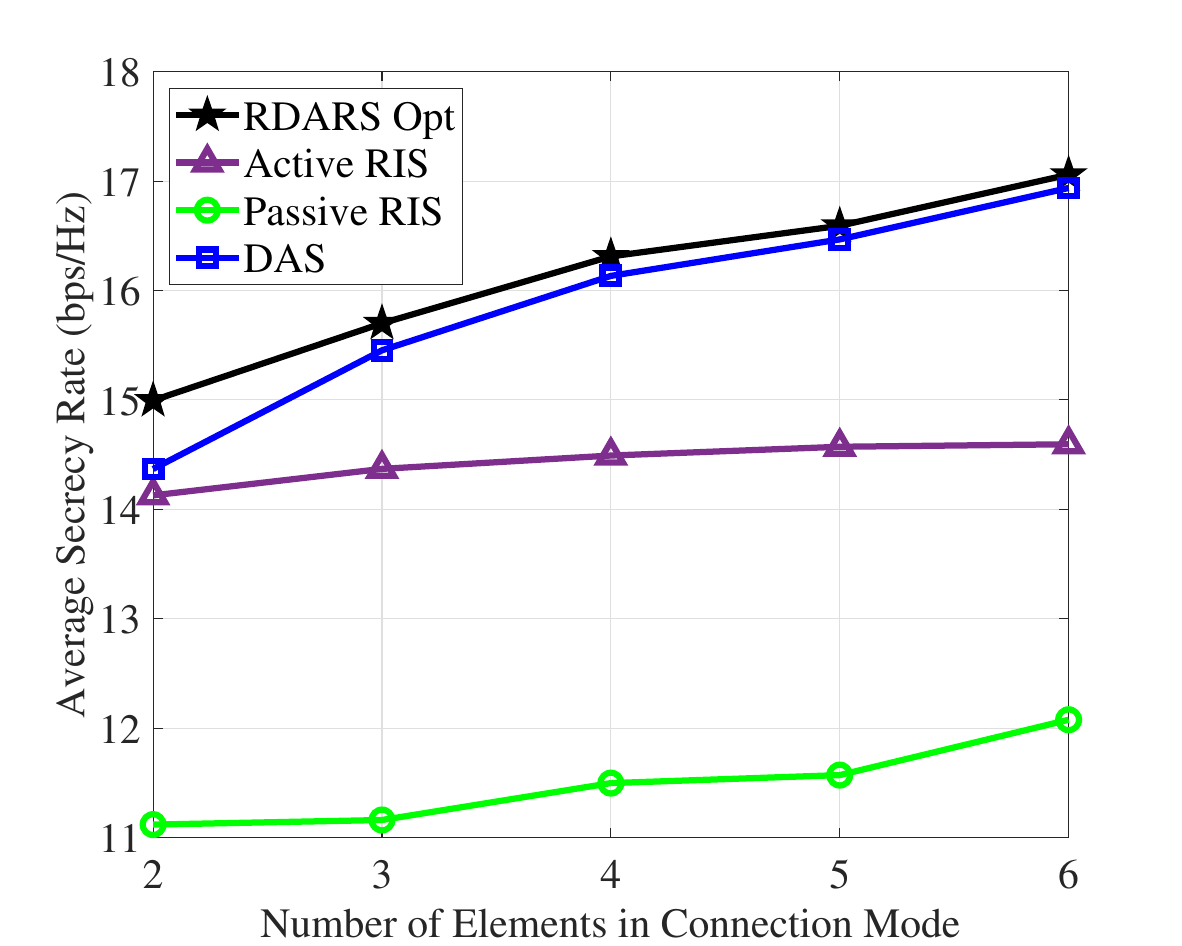}}
    \caption{Comparison of secrecy rate versus system parameters (a) Transmit Power $P$ (b) Number of Elements in RDARS $N$ (c) Number of Elements in connection mode $a$. }
    \vspace{-12pt}
    \label{fig_E1}
\end{figure*}

\vspace{-12pt}
\section{Conclusion}
In this letter, we proposed an innovative RDARS-driven secure communication system and investigated joint channel-aware mode selection and beamforming design to maximize the secrecy rate. Compared to conventional passive RIS-aided systems, the binary and cardinality constraints in RDARS further complicated the secrecy maximization problem. To address these challenges, we presented an efficient AO framework utilizing penalty-based FP and SCA methods. The simulations validated the potential of RDARS to enhance the secrecy rate and demonstrated its superiority over existing reflection surface-based schemes.

\bibliographystyle{IEEEtran}
\bibliography{RDARS_Secrecy}

\end{document}